\documentclass{article}
\usepackage[utf8]{inputenc}
\usepackage[T1]{fontenc}
\usepackage{amsmath}
\usepackage{dsfont}
\usepackage{fancyhdr}
\usepackage[table,xcdraw]{xcolor}
\usepackage{graphicx}
\usepackage{multirow}
\usepackage{endnotes}
\usepackage{float}
\usepackage{vmargin}
\usepackage{amssymb}
\usepackage{subfigure}
\usepackage{import}
\usepackage{stackrel}
\usepackage{wrapfig}
\usepackage{amsthm}
\usepackage{rotating}
\usepackage[title]{appendix}
\usepackage{listings}
\usepackage{soul}
\usepackage{stmaryrd}
\usepackage{mathtools}
\usepackage{tikz}
\usepackage{bbm}
\usepackage{mleftright}
\mleftright
\usetikzlibrary{matrix}
\usepackage{tikz}
\usetikzlibrary{shapes,arrows,chains, decorations.pathmorphing, decorations.pathreplacing}
\tikzset{quantum/.style={decorate, decoration=snake}}

\newcommand{\diagdots}[3][-25]{%
  \rotatebox{#1}{\makebox[0pt]{\makebox[#2]{\xleaders\hbox{$\cdot$\hskip#3}\hfill\kern0pt}}}%
}

\newcommand{\ket}[1]{\left|#1\right\rangle}		
\newcommand{\bra}[1]{\left\langle#1\right|}

\newcommand{\braket}[2]{\left\langle #1\lvert#2\right\rangle}
\newcommand{\abs}[1]{\lvert #1\rvert}

\newcommand{\norm}[1]{\| #1\|}
\newcommand{\eps}{\varepsilon}

\newcommand{\tr}[1]{\mathrm{Tr}\left[#1\right]}

\newcommand{\QPVcohf}{$\mathrm{QPV}_{\mathrm{coh}}^{f}$}
\newcommand{\QPVcohfm}{$\mathrm{QPV}_{\mathrm{coh}_{m_0}}^{f}$}

\newcommand{\QPVBB}{$\mathrm{QPV}_{\mathrm{BB84}}$}

\newcommand{\QPVBBf}{$\mathrm{QPV}_{\mathrm{BB84}}^{f}$}

\definecolor{darkred}{RGB}{179, 16, 32}

\usepackage{slashed}
\usepackage{enumitem}
\usepackage{authblk}

\usetikzlibrary{positioning, arrows.meta}

\theoremstyle{plain}
\newtheorem{theorem}{Theorem}[section]
\newtheorem{remark}[theorem]{Remark}
\newtheorem{prop}[theorem]{Proposition}
\newtheorem{definition}[theorem]{Definition}

\newtheorem{lemma}[theorem]{Lemma}

\newtheorem{corollary}[theorem]{Corollary}

\usepackage[normalem]{ulem}

\usetikzlibrary{decorations.pathreplacing,calligraphy}

\usepackage[pdfencoding=auto, psdextra]{hyperref}
\hypersetup{
    bookmarksnumbered=true, 
    unicode=false, 
    pdfstartview={FitH}, 
    pdftitle={Security of a Continuous-Variable based Quantum Position
Verification Protocol}, 
    pdfauthor={Rene Allerstorfer, Llorenc Escola-Farras, Arpan Akash Ray, Boris Skoric, and Florian Speelman}, 
    colorlinks=true, 
    allcolors=blue,  
}

\title{Continuous-variable Quantum Position Verification secure against entangled attackers}
\author[1]{Rene Allerstorfer}
\author[2]{Lloren\c{c} Escol\`a-Farr\`as}
\author[3]{Arpan Akash Ray}
\author[3]{Boris \v{S}kori\'{c}}
\author[2]{Florian Speelman}
\date{\today}

\affil[1]{QuSoft \& CWI Amsterdam, The Netherlands}
\affil[2]{QuSoft \& Informatics Institute, University of Amsterdam, The Netherlands }
\affil[3]{TU Eindhoven, The Netherlands }

\begin{document}

\maketitle
\begin{abstract}
\noindent Motivated by the fact that coherent states may offer practical advantages it was recently shown that a continuous-variable (CV) quantum position verification (QPV) protocol using coherent states could be securely implemented if and only if attackers do not pre-share any entanglement. In the discrete-variable (DV) analogue of that protocol it was shown that modifying how the classical input information is sent from the verifiers to the prover leads to a favourable scaling in the resource requirements for a quantum attack. In this work, we show that similar conclusions can be drawn for CV-QPV. By adding extra classical information of size $n$ to a CV-QPV protocol, we show that the protocol, which uses a coherent state and classical information, remains secure, even if the quantum information travels arbitrarily slow, against attackers who pre-share CV (entangled) states with a linear (in $n$) cutoff at the photon number. We show that the protocol remains secure for certain attenuation and  excess noise. 

\end{abstract}
\tableofcontents

\section{Introduction}
Position-based cryptography utilises the geographic location of a party as the only cryptographic credential to authenticate it, without further assumptions. For instance, to authenticate a website or online media one could verify that the server is located \textit{where} it should be or that the media was recorded \textit{where} it claims to have been (instead of being created by a powerful AI or untrusted parties, for example). Part of position-based cryptography is the task of position verification, where an untrusted prover aims to convince verifiers that he is present at a certain position $P$. 

This primitive was first introduced by Chandran, Goyal,
Moriarty, and Ostrovsky \cite{OriginalPositionBasedCryptChandran2009}, and it has been shown that no classical position-verification protocol can exist, due to a universal attack based on cloning input information. This attack fails in the quantum setting because of the no-cloning theorem \cite{Wootters1982NoCloning}. Quantum position verification (QPV) has been studied\footnote{under the name of `quantum tagging'} since the early 2000s by several authors \cite{PatentKentANdOthers,Malaney_2010_b,Malaney_2010_a,Lau_2011}. Proposed QPV protocols rely on both relativistic constraints (in a $d$-dimensional Minkowski space-time $M^{(d,1)}$) and the laws of quantum mechanics. In the literature, usually the case $d=1$ studied for simplicity, i.e.\ verifying the position of $P$ in a line (by two verifiers $V_0$ and $V_1$ who are placed on the left and right of $P$, respectively), since it makes the analysis easier and the main ideas generalize to higher dimensions. Despite the failure of the classical universal attack, a universal quantum attack has been found \cite{Buhrman_2014, Beigi_2011}. However, this attack consumes an amount of entanglement exponential in the input size and is therefore not practically feasible. Thus, we may still find secure QPV protocols in the bounded-entanglement model. 

The analysis of the entanglement resources needed turns out to be a deep question in its own right~\cite{Buhrman_2013_Graden_hose,speelman2016instantaneous,dolev2022non,bluhm2022single,Cree2023coderoutingnew,allerstorfer2023relating,apel2024security,asadi2024rank,asadi2024linear}. Many protocols have since been proposed~\cite{Chakraborty_2015,SWAP_protocol_Rene_et_all,gonzales2019bounds,liu_et_al:LIPIcs.ITCS.2022.100,allerstorfer2023security,allerstorfer2023making} and different security models have been studied~\cite{Unruh_2014_QPV_random_oracle,gao2016quantum,dolev2019constraining,allerstorfer2022role}. Recent work has focused on the practicality of implementing position-verification protocols. Aspects such as channel loss and error tolerance of certain QPV protocols must be taken into account~\cite{allerstorfer2022role,escolafarras2022singlequbit, allerstorfer2023making}. 

Most previous QPV protocols have been based on finite-dimensional quantum systems, with the exception of \cite{qi2015loss, allerstorfer2023security}.

Continuous-variable quantum systems are relevant for quantum communication and quantum-limited detection and imaging techniques because they provide a quantum description of the propagating electromagnetic field. Of particular relevance are the eigenstates of the annihilation operator, also known as coherent states, and their quadrature squeezed counterparts known as squeezed coherent states.
Much research regarding continuous-variable quantum key distribution (QKD) has been conducted. Firstly proposed with discrete \cite{PhysRevA.61.010303, PhysRevA.61.022309, PhysRevA.62.062308} and Gaussian \cite{PhysRevA.63.052311} encoding of squeezed states, soon a variety of protocols were published on Gaussian-modulated CV-QKD with coherent states \cite{PhysRevLett.88.057902, article_Grossmman_03, 10.5555/2011564.2011570, PhysRevLett.93.170504}. In this paper, we employ many techniques that are common in CV-QKD. Theoretical reviews with practical considerations of CV-QKD can be found in \cite{garcia2007quantum, leverrier:tel-00451021}.

CV systems are much simpler to handle in practice and leverage several decades of experience in
coherent optical communication technology. 
One particular advantage is that no true single-photon preparation or detection is necessary. 
Clean creation and detection of single photons is still expensive and technically challenging, especially if photon number resolution is desired. 
In contrast, homodyne and heterodyne measurements are easy to implement and much existing infrastructure is geared towards handling light at low-loss telecom wavelengths (1310nm, 1550nm), whereas an ideal single photon source in these wavelength bands still has to be discovered and frequency up-conversion is challenging and introduces new losses and errors.

In \cite{allerstorfer2023security}, the CV analogue of the \QPVBB~protocol, first introduced and studied in~\cite{OriginalQPV_Kent2011, Buhrman_2014, Lau_2011}, was defined and analysed. In this article, we extend the CV-QPV literature by considering the CV version of the practically interesting protocol \QPVBBf~\cite{bluhm2022single, escolafarras2022singlequbit}, where the classical input information is split up (into, say, $x,y$) and each verifier sends out one part of it. The prover then applies the appropriate measurement based on the value $f(x,y)$ for the chosen protocol function $f$. The advantage of this is that the required quantum resources for a successful attack become larger and scale linearly in the size $n$ of the \textit{classical} input strings $x,y$. Thus, increasing the classical input size makes the quantum attack harder -- a very favourable property of \QPVBBf. It is theoretically, and also potentially practically, interesting whether this property holds the same way in the CV case, which is why we study it.

Employing previous results from \cite{bluhm2022single} and \cite{allerstorfer2023security}, the main take-away of this work is that, indeed, the CV protocol shares the desired characteristics regarding entanglement attacks of the discrete variable version. 
However, it was shown in \cite{qi2015loss} that a simple generic attack exists as long as the transmission is $t \leq 1/2$ for this type of CV-QPV protocol.

More concretely, we show that, for a random function $f$, the protocol remains secure against attackers who pre-share a quantum resource state with dimension linear in $n$. Moreover, the protocol remains secure even if the quantum information is sent arbitrarily slowly. We also consider attenuation and excess noise in the CV channel in our analysis. The underlying technique is to lower bound the attackers' entropy about the quadrature value $R$ that the protocol asks for. To do so, we employ the continuity of the conditional entropy for continuous variables in terms of the energy as shown in \cite{Winter_2016}. This then implies that the sample the attackers respond with will necessarily have a higher variance than the honest sample, which allows the verifiers to distinguish between an honest and a dishonest sample.

\section{Preliminaries}

Let $\mathcal{H}$, $\mathcal{H'}$ be finite-dimensional Hilbert spaces, we denote by $\mathcal{B}(\mathcal{H},\mathcal{H'})$ the set of bounded operators from $\mathcal{H}$ to $\mathcal{H'}$ and $\mathcal{B}(\mathcal{H})=\mathcal{B}(\mathcal{H},\mathcal{H})$. Denote by $\mathcal{S}(\mathcal{H})$ the set of quantum states on $\mathcal{H}$,~i.e.\ $\mathcal{S}(\mathcal{H})=\{\rho\in\mathcal{B}(\mathcal{H})\mid \rho\geq0, \tr{\rho}=1)\}$. A pure state will be denoted by a ket $\ket{\psi}\in\mathcal{H}$. The trace distance between two quantum states $\rho$ and $\sigma$ is given by 
\begin{equation}
\frac{1}{2}\norm{\rho-\sigma}_1.
\end{equation}
We will write $\frac{1}{2}\norm{\ket{\psi_1}-\ket{\psi_2}}_1$
for pure states $\ket{\psi_1},\ket{\psi_2}$. The fidelity bewteen two quantum states $\rho$ and $\sigma$ is defined as
\begin{equation}
    F(\rho,\sigma):=\tr{\sqrt{\sqrt{\sigma}\rho\sqrt{\sigma}}},
\end{equation}
in particular, for pure states $\ket{\psi_1},\ket{\psi_2}$, $F(\ket{\psi_1},\ket{\psi_2})=\abs{\braket{\psi_1}{\psi_2}}$. The purified distance for quantum states $\rho$ and $\sigma$ is defined as
\begin{equation}
    \mathcal{P}(\rho,\sigma):=\sqrt{1-F(\rho,\sigma)^2}.
\end{equation}

\begin{definition} Let $X$ be a continuous random variable with probability density function $f(x)$, and let $\mathcal{X}$ be its support set. The \emph{differential Shannon entropy} $h(X)$ is defined as
\begin{equation}
    h(X)=-\int_{\mathcal{X}}f(x)\log f(x) \, \mathrm{d}x,
\end{equation}
    where, if not otherwise mentioned, we use $log$ in base 2. 
\end{definition}
\begin{lemma}
\label{lemma:hscaling}
Let $\beta >0$ and $X\in\mathbb R$. It holds that
$h(\beta X) = h(X)+\log\beta$.
\end{lemma}

As introduced in \cite{furrer2014position}, let $\rho_{AB}$ be a bipartite state on systems $A$ and $B$, which correspond to a system to be measured and a system held by an observer. Let $X$ be a continuous random variable, $\alpha=2^{-n}$ for some $n\in\mathbb N$, and consider the intervals $\mathcal{I}_{k;\alpha}:=(k\alpha,(k+1)\alpha]$ for $k\in\mathbb Z$.  Here $\rho_B^{k;\alpha}$ denotes the sub-normalized density matrix in $B$ when $x$ is measured in $\mathcal{I}_{k;\alpha}$, $\rho_B^x$ denotes the conditional reduced density matrix in $B$ so that $\int_{\mathcal{I}_{k;\alpha}}\rho_B^xdx=\rho_B^{k;\alpha}$, and $Q_{\alpha}$ denotes the random variable that indicates which interval $x$ belongs to. These notions are used in the continuous version of the conditional entropy.

\begin{definition}
    The \emph{quantum conditional von Neumann entropy} is defined as 
    \begin{equation}
        H(Q_{\alpha}|B)_\rho:=-\sum_{k\in\mathbb Z}D(\rho^{k;\alpha}_{B}||\rho_B).
    \end{equation}
\end{definition}

\begin{definition} The  \emph{differential quantum conditional von Neumann entropy} is defined as
\begin{equation}
    h(X|B)_{\rho} := -\int_{\mathbb R} D(\rho_B^x || \rho_B) \, \mathrm{d}x. 
\end{equation}
\end{definition}
    
\noindent The basis of our security proofs is the quantum-mechanical uncertainty principle. We use the following form for the differential entropy in a tripartite setting of a guessing game, as is often useful in the context of quantum cryptography. 

\begin{lemma} \label{lemma:uncertainty}\cite{furrer2014position} Let $\rho_{ABC}$ be a tripartite density matrix on systems $A$, $B$ and $C$. Let $Q$ and $P$ denote the random variables of position and momentum respectively, resulting from a homodyne measurement on the $A$ system and let the following hold: $h(Q|B)_{\rho}, h(P|C)_{\rho}>-\infty$ and $H(Q_{\alpha}|B)_{\rho}, H(P_{\alpha}|C)_{\rho}<\infty$ for any $\alpha>0$. Then\begin{equation}    h(Q|B)_{\rho}+h(P|C)_{\rho}\geq\log(2\pi).\end{equation}\end{lemma}

We will make use of a type of Alicki-Fannes \cite{Alicki_2004} inequality for continuity of the conditional entropy for continuous variables in terms of the energy as shown in \cite{Winter_2016}.  Consider the Hamiltonian on a system $A$ being the harmonic oscillator with 
\begin{equation}
    H=\hbar \omega \hat{N},
\end{equation}
with the unusual energy convention that the ground state has energy $0$ instead of $\frac{1}{2}\hbar \omega$. Throughout the paper, we will consider units such that $\hbar \omega =1$ (note that we will consider a fixed wavelength for an input state, see below).

\begin{lemma} \label{lemma:close_states_close_entropy}(Lemma 18, \cite{Winter_2016}) Let $\alpha\in[0,\frac{1}{2}]$. Consider a Hamiltonian $H=H_A\otimes \mathbb I_B$, with system $A$ composed of one harmonic oscillator and arbitrary system $B$. Let there be states $\rho$ and $\sigma$ on the bipartite system $\mathcal{H}_A\otimes\mathcal H_B$ with $\tr{\rho H},\tr{\sigma H}\leq E$. If $\frac{1}{2}\norm{\rho-\sigma}_1\leq \Tilde{\varepsilon}$, then
    \begin{equation}
        \abs{h(A|B)_{\rho}-h(A|B)_{\sigma}}\leq
        \left(\frac{1+\alpha}{1-\alpha}+2\alpha\right)
        \left[2\Tilde\varepsilon \left(\log({E}+1)+\log\frac{e}{\alpha(1-\Tilde\varepsilon)}\right)+6\Tilde{h}\left(\frac{1+\alpha}{1-\alpha}\Tilde{\varepsilon}\right)\right],
    \end{equation}
    where
    \begin{equation}
    \Tilde{h}(x):=
        \begin{cases}
            -x\log x-(1-x)\log(1-x) &\text{ if } x\leq\frac{1}{2}\\
            1 &\text{ if } x\geq \frac{1}{2}.
        \end{cases}
    \end{equation}
\end{lemma}

Furthermore, we will make use of the following estimation inequality.

\begin{theorem} \cite{cover1999elements}\label{theorem:fano}
    Let $X$ be a random variable and $\hat{X}(Y)$ an estimator of $X$ given side information $Y$, then
    \begin{equation}
        \mathbb{E} \left[ \left(X-\hat{X}(Y)\right)^2 \right] \geq \frac{1}{2\pi e}e^{2h_{\mathrm{nats}}(X|Y)},
    \end{equation}
    where $h_{\mathrm{nats}}(X|Y)$ is the conditional entropy in natural units. Moreover, if $X$ is Gaussian and $\hat{X}(Y)$ is its mean, then equality holds. 
\end{theorem}

\section{The \texorpdfstring{\QPVcohf}{QPVcohf}~protocol}
Based on the ideas in \cite{Buhrman_2013_Graden_hose,Unruh_2014_QPV_random_oracle,bluhm2022single},  
we introduce a variation of the quantum position verification protocol studied in \cite{allerstorfer2023security}. Instead of only a single verifier sending classical information, both verifiers send classical information that, combined, determine the action of the prover. When sent through a channel, a continuous-variable state gets attenuated and acquires excess noise. We will denote by $t\in[0,1]$ the attenuation parameter, and by  $u\geq 0$  the excess noise power of the quantum channel connecting $V_0$ and $P$. 
Then, the protocol is described as follows:

\begin{definition} \label{def qpv bb84 f} Let $n\in\mathbb{N}$ and  consider a $2n$-bit boolean function $f:\{0,1\}^n \times \{0,1\}^n \to  \{0,1\}$.  A round of the \QPVcohf~protocol is described as follows.
\begin{enumerate}
\item The verifiers $V_0$ and $V_1$ randomly choose bit strings $x,y\in\{0,1\}^n$, respectively. They draw two random variables $(r,r^\perp)$ from the Gaussian distribution ${\cal N}_{0,\sigma^2}$, for $\sigma\gg1$, and compute $f(x,y)$.  Verifier $V_0$ prepares a coherent state $\ket\psi$ 
with quadratures $(x_0,p_0) = (r\cos\theta + r^\perp\sin\theta, \; r\sin\theta-r^\perp\cos\theta)$, where $\theta=0$ if $f(x,y)=0$ and $\theta=\frac{\pi}{2}$ if $f(x,y)=1$.
\item The verifier $V_0$ sends $\ket\psi$ and $x$ to $P$, and the verifier $V_1$ sends $y$ to $P$ such that all information arrives at $P$ simultaneously. The classical information is required to travel at the speed of light whereas the quantum information can be arbitrarily slow. 

\item Immediately, $P$ computes $f(x,y)$ and 
performs a homodyne measurement on $\ket\psi$ in the direction $\theta=0$ if $f(x,y)=0$ or $\theta=\frac{\pi}{2}$ if $f(x,y)=1$, resulting in a value $r'\in{\mathbb R}$.
The prover broadcasts $r'$ to both verifiers at the speed of light.

\end{enumerate}
After $N$ rounds, the verifiers have received a sample of responses, which we denote as $(r_i')_{i=1}^N$. The verifiers check whether all prover responses arrived at the correct time,
and whether the reported values $(r_i')_{i=1}^N$  satisfy
\begin{equation}
    \frac1N\sum_{i=1}^N \frac{\left(r_i'-r_i\sqrt t\right)^2}{\frac12+u} <  \gamma, \quad \mbox{ with }  \quad
    \gamma\stackrel{\rm def}{=} 1 +\frac2{\sqrt N}\sqrt{\ln\frac1{\eps_{\rm hon}}} + \frac{2}{N}\ln\frac1{\eps_{\rm hon}}.
\label{score}
\end{equation}
Here $\eps_{\rm hon}$ is an upper bound on the honest prover's failure probability (as a parameter of the protocol). See Fig.~\ref{fig:protocol} for a schematic representation of the \QPVcohf~protocol. 

\end{definition}

\begin{figure}[h]
    \centering
    \begin{tikzpicture}[node distance=3cm, auto]
    \node (A) {$V_0$};
    \node [left=1cm of A] {};
    \node [right=of A] (B) {};
    \node [right=of B] (C) {$V_1$};
    \node [right=1cm of C] {};
    \node [below=of A] (D) {};
    \node [below=of B] (E) {P}; 
    
    \node [right=1.2cm of A] (M) {};
    
    \node [below=of C] (F) {};
    \node [below=of D] (G) {$V_0$};
    \node [below=of E] (H) {};
    \node [below=of F] (I) {$V_1$};
    \node [left= 6cm of E] (J) {};
    \node [below= 3cm of J] (K) {};
    \node [above= 3cm of J] (L) {};

    \node [right=1.05cm of A, transform canvas={xshift=+0pt,yshift=+6pt}] {$\diagdots[117]{2em}{0.1em}$};

    \draw [->, transform canvas={xshift=0pt, yshift = 0 pt}, quantum] (M) -- (E) node[midway] (x) {} ;
    \draw [->] (A) -- (E);
    \draw [->] (C) -- (E);
    \draw [->] (E) -- (I) node[midway] (q) {$r'$};
    \draw [->] (E) -- (G);

    \draw [->] (L) -- (K) node[midway] {time};

    \node[left=0.3cm of x, transform canvas={xshift=+ 2pt, yshift = +2 pt}] {$\ket\psi$};
    \node[left=1.5cm of x, transform canvas={xshift=+ 2pt, yshift = +2 pt}] {$x\in\{0,1\}^n$};
    \node[right = 2.9cm of x, transform canvas={xshift=+ 2pt, yshift = +2 pt}] {$y \in \{0,1\}^n$};
    \node[left = 3.3cm of q] {$r'$};
\end{tikzpicture}
\caption{One round of the \QPVcohf~protocol. The coherent state $\ket\psi$ originates from $V_0$ in the past.}
\label{fig:protocol}
\end{figure}
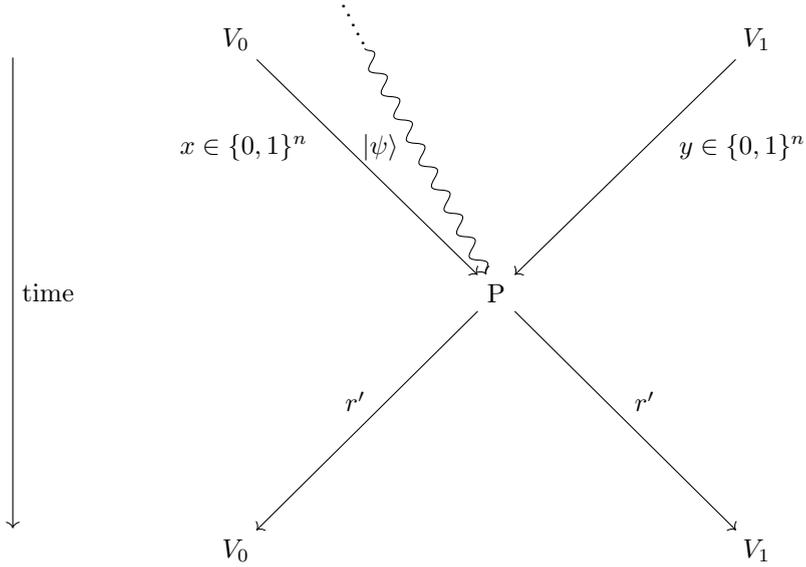
The excess noise can be modeled as $u=u_0\sigma^2$, with, for instance, a reasonable $u_0=0.01$, due to the prevalence of phase noise \cite{noisemodels}. The protocol becomes insecure for $u>0.25$, and thus the constant $u_0$ places a practical upper limit on the modulation variance $\sigma$.

The protocol in \cite{allerstorfer2023security} is as in Definition~\ref{def qpv bb84 f} but instead of sending $x$ and $y$, $V_1$ directly sends $\theta$. For the honest party, the only difference is that he has to compute $f(x,y)$ to determine $\theta$. 
We will use the following notation, for $\theta\in\{0,\frac{\pi}{2}\}$ we will denote by $\bar\theta$ the remaining value, i.e.\ if $\theta=0$, then $\bar\theta=\frac{\pi}{2}$ and if    $\theta=\frac{\pi}{2}$, then $\bar\theta=0$.

The honest prover's uncertainty about the displacements $r$, drawn by the random variable denoted by $R$ in each round, conditioned on his measurement outcomes~$r'$, drawn from the random variable denoted by $R'$ in each round, is the same as in the protocol where $\theta$ is sent directly from $V_1$, since the only change for $P$ is to compute $f(x,y)$ to determine $\theta$. This uncertainty was calculated in \cite{allerstorfer2023security} to be
\begin{equation}
\label{entropy prover}
    h(R|R')_{\psi} = \frac{1}{2} \log 2\pi e\Sigma^2,
\end{equation}
where
\begin{equation}
    \Sigma^2 = \left(\frac{1}{\sigma^2}+\frac{t}{1/2+u}\right)^{-1}.
\end{equation}
It is well known that the preparation of a coherent state with Gaussian distributed displacements $x_0,p_0\sim {\cal N}_{0,\sigma^2}$ is equivalent to preparing a two-mode squeezed state with squeezing parameter $\sinh^{-1}\sigma$ and then performing
a heterodyne $(\hat x,\hat p)$ measurement on one mode, with measurement outcome 
$\frac{(x_0, -p_0)}{\sqrt2 \tanh\sinh^{-1}\sigma}$. For notational brevity, we set $\lambda= \tanh\sinh^{-1}\sigma$.

In the purified version of \QPVcohf, the verifier $V_0$ prepares the two-mode squeezed state $\ket{\Psi}_{V_0P}$  with the above mentioned $\lambda$-dependent squeezing, which is given by
\begin{equation}
 \ket{\Psi}=\sqrt{1-\lambda^2}\sum_{m=0}^{\infty}\lambda^m\ket{mm},
\end{equation}
in the Fock space. Notice that $\lambda<1$. The verifier $V_0$ performs a heterodyne measurement with quadratures  rotated by an angle~$\theta$ on their register.
The measurement outcomes are $r/(\sqrt2 \lambda )$ and $-r^\perp/(\sqrt2 \lambda)$, resulting in displacement
$(r,r^\perp)$ in the state sent to the prover $P$. $P$ then performs a homodyne measurement 
under angle $\theta$ to recover $r$, as in the original protocol. 

$V_0$'s heterodyne measurement can be described as a double-homodyne measurement. First $V_0$ mixes its own mode with the vacuum using a beamsplitter, resulting in a two-mode state. 
On one of these modes, $V_0$ then performs a homodyne measurement in the $\theta$-direction, on the other mode in the $\theta+\frac\pi2$ direction.

In \cite{allerstorfer2023security} it is shown that the honest prover's uncertainty about $R$ evaluates to
\begin{equation}
    h(R|P)_{\Psi} = \frac{1}{2}\log \frac{\pi e (1+2u)}{t} + O\left(\frac1\sigma\right).
\end{equation}
Let $U=R/(\lambda\sqrt2)$ be the displacement in the $\theta$ direction as measured by $V_0$.

Then 
\begin{equation}\begin{split}
     h(U|P)_{\Psi} &= h\left(\frac{R}{\sqrt2 \lambda}\Big|P\right)_{\Psi} = h(R|P)_{\Psi}-\log (\sqrt2 \lambda)\\&=     \frac{1}{2}\log \frac{\pi e (1+2u)}{t} + O\left(\frac1\sigma\right)-\log (\sqrt2 \lambda).
\end{split}
\end{equation}
In the regime $\sigma\gg1$ (i.e. $\lambda \to 1$),

\begin{equation}\label{equ:h_honest}
    h(U|P)_{\Psi}\rightarrow  h(U|P)_{\psi}=\frac{1}{2}\log \frac{\pi e (1+2u)}{2t}.
\end{equation}
Unless stated otherwise, we will work in the regime $\sigma \gg 1$. Recall that $\sigma$ is in control of the verifiers.

\section{Security against bounded entanglement}
In this section, we prove security of the \QPVcohf~protocol, showing that with high probability, attackers who possess CV entangled states with a cutoff at photon number linear in $n$ will not be able to attack the protocol. 

To do so, we consider an `imaginary world' where the \QPVcohf~protocol, instead of using the state $\ket\Psi$, is executed with a cutoff at photon number $2^{m_0}$ using the state $\ket{\Psi_{m_0}}$,  given by
\begin{equation}
    \ket{\Psi_{m_0}}=\sqrt{\frac{1-\lambda^2}{1-(\lambda^2)^{2^{m_{0}}}  }}  \sum_{m=0}^{2^{m_0}-1}\lambda^m\ket{mm}.
\end{equation}
We will denote this variation of the protocol by \QPVcohfm. The state $\ket{\Psi_{m_0}}$ is an approximation of the state $\ket{\Psi}$ and can be made arbitrary close to it by increasing $m_0$. Note that 
\begin{equation}
    \mathcal{P}(\ket\Psi,\ket{\Psi_{m_0}})=\lambda^{2^{m_0}},
\end{equation}
i.e. $\ket{\Psi_{m_0}}$ is double exponentially close (in $m_0$) to $\ket{\Psi}$ (recall that $\lambda<1$). If one replaces the state $\ket{\Psi}$ by $\ket{\Psi_{m_0}}$, the probability that the verifiers accept the action of an honest party will change with probability at most $O(\lambda^{2^{m_0}})$. The cutoff reduces the dimension of the Hilbert space from infinite to $2^{m_0}$, which is the dimension of an $m_0$-qubit state space.

The energy of the $V_0$ subsystem is given by
\begin{equation}\label{energy of V}
   \bra{\Psi_{m_0}} H_{V_0}\ket{\Psi_{m_0}}_{V_0P}=\frac{\lambda ^2+\left(2^{m_0}-1\right) \lambda ^{2^{m_0+1}+2}-2^{m_0} \lambda
   ^{2^{m_0+1}}}{\left(\lambda ^2-1\right) \left(\lambda ^{2^{m_0+1}}-1\right)}=\frac{2^{m_0} \left(\frac{\sigma ^2}{\sigma ^2+1}\right)^{2^{m_0}}}{\left(\frac{\sigma
   ^2}{\sigma ^2+1}\right)^{2^{m_0}}-1}+\sigma ^2,
\end{equation}
which tends to $\sigma^2$  (the expected energy of the challenge state chosen by the verifiers) 
as $m_0$ tends to infinity.

The most general attack to \QPVcohfm~ for adversaries with a photon-number cutoff such that their Hilbert space is isomorphic to a multi-qubit Hilbert space, consists of an adversary Alice between $V_0$ and $P$, and an adversary Bob between $V_1$ and $P$. They proceed as follows:
\begin{enumerate}
    \item \emph{Preparing:} The attackers prepare a joint (possibly entangled) CV state with a cutoff at the photon number of $q$ qubits each.  
    
    \item \emph{Intercepting:} Alice intercepts the quantum information sent from $V_0$. At this stage, $V_0$, Alice and Bob share a state $\ket{\chi}_{VPAA_cBB_c}$ of dimension $2^{2q+2m_0}$. 
    Here $V$ is the register kept by $V_0$, and $P$ is the challenge register that $V_0$ sends.
    Alice controls the registers $P$, $A$ and $A_c$, and Bob possesses the registers $B$ and $B_c$. 
    Moreover,  Alice and Bob intercept $x$ and $y$ and perform arbitrary quantum channels depending on the intercepted classical information: $U^x_{PAA_c}$ and $V^y_{BB_c}$, respectively, ending up with the state  $\ket{\phi}_{VPAA_cBB_c}$.
     
    \item \emph{Communicating:} Alice and Bob send a copy of $x$ and $y$ to the other attacker, respectively. Alice keeps registers $P$ and $A$ and sends register $A_c$ to Bob, and Bob keeps register $B$ and sends $B_c$ to Alice. 

    \item \emph{Measuring:} Upon receiving the information sent by the other party, Alice and Bob locally apply arbitrary POVMs $\{A^{xy}_{PAB_c}\}$ and $\{B^{xy}_{A_cB}\}$ to obtain classical answers, which will be sent to their closest verifier, respectively.
\end{enumerate}
Due to the Stinespring dilation, we can consider the quantum channels to be unitaries.

\begin{figure}[h]
    \centering
    \begin{tikzpicture}[node distance=4cm, auto]
    \node (A) {Alice};
    \node [right= of A] (B) {Bob};
    
    \node [below = 1cm of A ] (A_Intercepts) {$U^x_{PAA_c}$};
    \node [below = of A_Intercepts ] (A_answers) {$A^{xy}_{PAB_c}$};
    \node [right = 2cm of A ](middle){};
    \node [below = 0.8cm of middle ](middle2){$\ket{\chi}_{VPAA_cBB_c}$};
    \node [below = 1.4cm of middle ](middle2){$\ket{\phi}_{VPAA_cBB_c}$};
    \node [below = 1cm of A_answers ] (A_final) {};
    
    \node [below = 1cm of B ] (B_Intercepts) {$V^y_{BB_c}$};
    \node [below = of B_Intercepts ] (B_answers) {$B^{xy}_{A_cB}$};
    \node [below = 1cm of B_answers ] (B_final) {};

    \draw[dotted][-](A_Intercepts)--(B_Intercepts){};
    \node [left = 2cm of A] (t0){};
    \node [below = 7cm of t0] (t1){};
    \draw [->] (t0) -- (t1) node[midway] {time};
    \draw [->, transform canvas={xshift=0pt, yshift = 0 pt}, quantum] (A_Intercepts) -- (B_answers) node[midway] (x) {};
    \draw [->, transform canvas={xshift=0pt, yshift = 0 pt}, quantum] (B_Intercepts) -- (A_answers) node[midway] (x) {};

\end{tikzpicture}
\caption{Schematic depiction of a generic attack on \QPVcohfm.}
\label{fig:attack}
\end{figure}
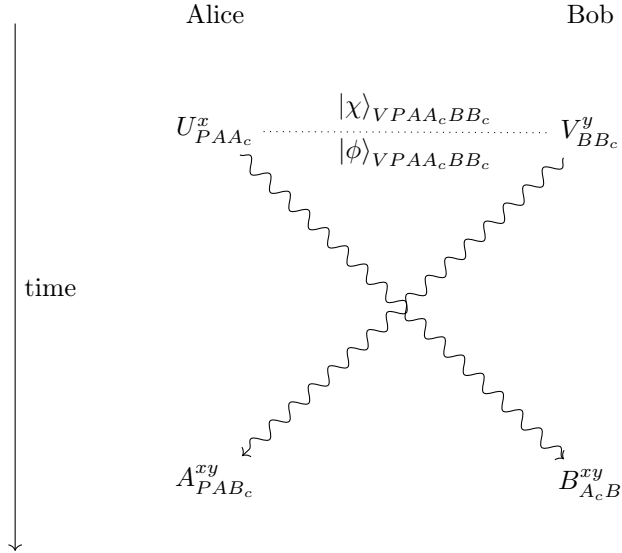

\begin{definition} The tuple $\{\ket{\chi}_{VPAA_cBB_c},U^x_{PAA_c},V^y_{BB_c},A^{xy}_{PAB_c},B^{xy}_{A_cB}\}_{xy}$ is a \emph{$q$-qubit strategy for \QPVcohfm}. Moreover, we say that a $q$-qubit strategy for \QPVcohfm~is $(\varepsilon,l)$-perfect if for $l$ pairs of strings $(x,y)$, for $\theta\in\{0,\frac{\pi}{2}\}$,
\begin{equation}\label{eq:epsilon-perfect_strategy}
    h(U_{\theta}|PAB_c)_{\phi}\leq h(U|P)_{\psi}+\varepsilon \quad \text{and} \quad h(U_{\theta}| A_cB)_{\phi}\leq h(U|P)_{\psi}+\varepsilon.
\end{equation}  
\end{definition}
Notice that there is no $\theta$-dependence on the right-hand side of the inequality since $ h(U|P)_{\psi}$ does not depend on $\theta$. The parameter $\varepsilon$, see analysis below, will quantify the difference between the uncertainty of the honest prover and attackers. 
It will have to be picked depending on the number of rounds that the protocol is run sequentially. 
Next, we define `good' states to attack the protocol for either $\theta=0$ or $\theta=\pi/2$. 
We will see that a good state for the $\theta=0$ case cannot be too close (in trace distance) to a good state for the $\theta=\pi/2$ case. This restricts the attackers. 

\begin{remark}
   Notice that we consider strategies starting with the state $\ket{\chi}_{VPAA_cBB_c}$ instead of $\ket{\Psi_{m_0}}_{VP}\otimes\ket{\chi}_{AA_cBB_c}$. This will give more power to the attackers, but it will include the fact that the quantum information sent from $V_0$ can travel arbitrarily slow, and the attackers are allowed to modify $\ket{\Psi_{m_0}}_{VP}\otimes\ket{\chi}_{AA_cBB_c}$ to end up with any arbitrary state $\ket{\chi}_{VPAA_cBB_c}$.
\end{remark}

\begin{definition}
    Let $\varepsilon\geq 0$, and let $q$ be the number of qubits that Alice and Bob each hold at the \emph{preparing} stage of a multi-qubit attack on \QPVcohfm. We define $\mathcal{S}_{\theta}^{\varepsilon}$ as
    \begin{equation}
    \begin{split}
        \mathcal{S}_{\theta}^{\varepsilon}:=\{\ket{\phi}_{VPAA_cBB_c}\in\mathbb C^{2q+2m_0}\mid \exists\textrm{ POVMs }\{A^{xy}_{PAB_c}\},\{B^{xy}_{A_cB}\} \text{ s.t. \eqref{eq:epsilon-perfect_strategy} holds}\}.
        \end{split}
    \end{equation}
   
\end{definition}

\begin{prop}
Let $\varepsilon>0$, and  $\ket{\phi_{0}}_{VPAA_cBB_c}\in \mathcal{S}_{0}^{\varepsilon}$, and $\ket{\phi_{\frac{\pi}{2}}}_{VPAA_cBB_c}\in \mathcal{S}_{\frac{\pi}{2}}^{\varepsilon}$, with bounded energies $\tr{\rho^0H},\tr{\rho^1H}\leq E$, where $\rho^0$ and $\rho^1$ are the respective density matrices of $\ket{\phi_{0}}$ and $\ket{\phi_{\frac{\pi}{2}}}$. The corresponding Hamiltonian is the harmonic oscillator on system $V$ and identity on the other systems. Let $\frac{1}{2}\geq\alpha\geq0$ and $\Tilde\varepsilon>0$ be such that

\begin{equation}\label{eq:condition_alpha_epsilon}
\varepsilon<\frac{1}{2}\log\frac{4t}{e(1+2u)}-\left(\frac{1+\alpha}{2(1-\alpha)}+\alpha\right)
        \left[2\Tilde\varepsilon \left(\log(E     +1)+\log\frac{e}{\alpha(1-\Tilde\varepsilon)}\right)+6\Tilde{h}\left(\frac{1+\alpha}{1-\alpha}\Tilde{\varepsilon}\right)\right].
\end{equation}
 Then, 
\begin{equation}
    \frac{1}{2}\lVert \ket{\phi_{0}}-\ket{\phi_{\frac{\pi}{2}}}\rVert_1 > \Tilde\varepsilon.
\end{equation}
\end{prop}

Notice that the energy bound $E$ is the energy given by the Hamiltonian corresponding to one harmonic oscillator in the $V$ system (\ref{energy of V}), which approaches $\sigma^2$ from below. Moreover, from \eqref{eq:condition_alpha_epsilon}, we see that the $\varepsilon$ will need to be picked taking a value at most $ \frac{1}{2}\log\frac{4t}{e(1+2u)}$. In order to have non-negative $\Tilde{\varepsilon}$, we need 
\begin{equation}\label{eq:relation_t_u}
    4t>e(1+2u),
\end{equation}
see Fig.~\ref{fig:u_t_relation}. The maximum value of ${\varepsilon}$ will be upper bounded by 
\begin{equation}
    \varepsilon<\frac{1}{2}\log\frac{4t}{e(1+2u)}\leq\frac{1}{2}\log\frac{4}{e}\simeq0.278652,
\end{equation}
since the maximum value is reached by  $t=1$ and $u=0$. 
\begin{figure}
    \centering
    \includegraphics[width=80mm]{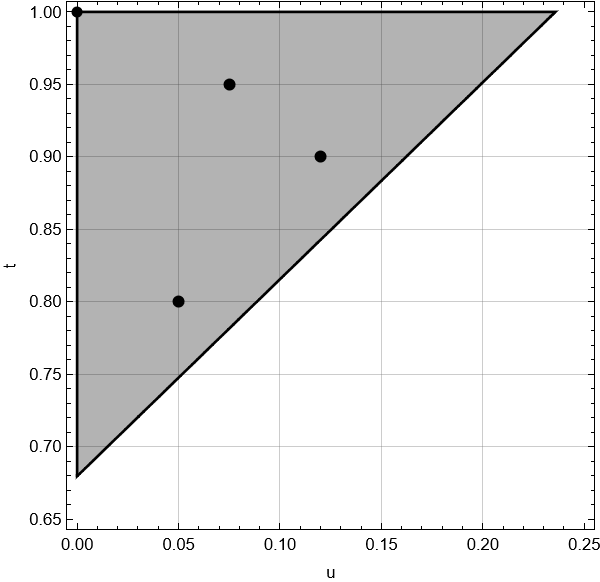}
    \caption{Necessary condition for $u$ and $t$ so that \eqref{eq:condition_alpha_epsilon} is fulfilled (gray region). The blue dots are specific $(u,t)$ for which we do numerical analysis in Section~\ref{section:parameters}.}
    \label{fig:u_t_relation}
\end{figure}

\begin{proof} Let $\rho^\theta$ and $\rho^{\Bar\theta}$ be the density matrices of $\ket{\phi_{\theta}}_{VPAA_cBB_c}$ and $\ket{\phi_{\Bar{\theta}}}_{VPAA_cBB_c}$, respectively.  By hypothesis, 
    \begin{equation}\label{eq:thetatheta}  h(U_{\theta}|PAB_c)_{\rho^\theta}\leq h(U|P)_{\psi}+\varepsilon,\text{ and }      
    h(U_{\Bar\theta}|PAB_c)_{\rho^{\Bar\theta}}\leq h(U|P)_{\psi}+\varepsilon. 
    \end{equation}
By Lemma \ref{lemma:uncertainty}, 
\begin{equation}
h(U_{\theta}|PAB_c)_{\rho^\theta}+h(U_{\Bar\theta}|A_cB)_{\rho^\theta}\geq\log2\pi.
\end{equation}
Then, 
\begin{equation}
    h(U_{\Bar\theta}|A_cB)_{\rho^\theta}\geq \log2\pi-h(U_{\theta}|PAB_c)_{\rho^\theta}\geq \log2\pi-h(U|P)_{\psi}-\varepsilon,
\end{equation}
where in the last inequality we used \eqref{eq:thetatheta}. Therefore, 
\begin{equation}
     h(U_{\Bar\theta}|A_cB)_{\rho^\theta}-h(U_{\Bar\theta}|PAB_c)_{\rho^{\Bar\theta}}\geq \log2\pi-2h(U|P)_{\psi}-2\varepsilon.
\end{equation}
In the regime $\sigma\gg1$, $h(U|P)_{\psi}\rightarrow  \frac{1}{2}\log \frac{\pi e (1+2u)}{2t}$, and thus, 
\begin{equation}
     h(U_{\Bar\theta}|A_cB)_{\rho^\theta}-h(U_{\Bar\theta}|PAB_c)_{\rho^{\Bar\theta}}\geq \log\frac{4t}{e(1+2u)}-2\varepsilon>0,
\end{equation}
where the last inequality comes from the fact that, by hypothesis, $\frac{1}{2}\log\frac{4t}{e(1+2u)} > \varepsilon$. This leads to 
\begin{equation}
     \abs{h(U_{\Bar\theta}|A_cB)_{\rho^\theta}-h(U_{\Bar\theta}|PAB_c)_{\rho^{\Bar\theta}}}\geq \log\frac{4t}{e(1+2u)}-2\varepsilon,
\end{equation}
By hypothesis, 
$\log\frac{4t}{e(1+2u)}-2\varepsilon>\left(\frac{1+\alpha}{1-\alpha}+2\alpha\right)
        \left[2\Tilde\varepsilon \left(\log(E     +1)+\log\frac{e}{\alpha(1-\Tilde\varepsilon)}\right)+6\Tilde{h}\left(\frac{1+\alpha}{1-\alpha}\Tilde{\varepsilon}\right)\right]$. Thus, by the contrapositive of Lemma \ref{lemma:close_states_close_entropy}, 
\begin{equation}
    \frac{1}{2}\lVert\rho^\theta-\rho^{\Bar\theta}\rVert_1 > \Tilde\varepsilon.
\end{equation}

\end{proof}

\begin{definition} \cite{bluhm2022single} Let $q,k,n\in\mathbb N$, $\varepsilon>0$. Then, $g:\{0,1\}^{3k}\rightarrow\{0,1\}$ is an $(\varepsilon,1)$-classical rounding of size $k$ if for all $f:\{0,1\}^{2n}\rightarrow\{0,1\}$, for all states $\ket\psi$ on $2q+2m_0$ qubits, for all $l\in\{1,...,2^{2n}\}$ and for all $(\varepsilon,l)$-perfect $q$-qubit strategies for \QPVcohfm, there are functions $f_A:\{0,1\}^n\rightarrow\{0,1\}^k$, $f_B:\{0,1\}^n\rightarrow\{0,1\}^k$ and $\gamma\in\{0,1\}^k$ such that $g(f_A(x),f_B(y),\gamma)=f(x,y)$ on at least $l$ pairs $(x,y)$. 

\begin{lemma} \cite{book_prob_banach_spaces}
\label{lemma size delta-net} 
Let $|||*|||$ be any norm on $\mathbb{R}^{n_0}$, for $n_0\in\mathbb{N}$. There is a $\delta$-net $S$ of the unit sphere of $(\mathbb{R}^{n_0},|||*|||)$ of cardinality at most $(1+2/\delta)^{n_0}$.
\end{lemma}

\end{definition}
\begin{lemma}
\label{lemma unit vectors trace distance norm}
\cite{bluhm2022single} Let $\ket{x},\ket{y}\in\mathbb{C}^d$, for $d\in\mathbb N$, be two unit vectors. Then, $\mathcal{P}(\ket{x},\ket{y})\leq \| \ket{x}-\ket{y}\|_2$.
\end{lemma}

\begin{prop}\label{prop:size_k}
    Let $\varepsilon$, $\Tilde{\varepsilon}$ be such that if $\ket{\varphi_{\theta}}\in\mathcal{S}^{\varepsilon}_{\theta}$ and $\ket{\varphi_{\Bar\theta}}\in\mathcal{S}^{\varepsilon}_{\Bar\theta}$ implies $\mathcal{P}(\ket{\varphi_{\theta}},\ket{\varphi_{\Bar\theta}})>\Tilde\varepsilon$, then there is an $(\varepsilon,q)$-classical rounding of size $k=2^{2q+2m_0}\log\big(1+\frac{4}{\sqrt[3]{4(2+\Tilde\varepsilon)}-2}\big)$.
\end{prop}

\begin{proof}
We follow the same techniques as in the proof of Lemma~3.12 in \cite{bluhm2022single}. Let $\delta<\sqrt[3]{\frac{2+\Tilde{\varepsilon}}{2}}-1$, and consider $\delta$-nets $\mathcal{N}_S$, $\mathcal{N}_A$ and $\mathcal{N}_B$, where the first is for the set of pure states on $2q+2m_0$ qubits in Euclidean norm and the other nets are for the set of unitaries in dimension $2^q$ in operator norm. They are such that $\abs{\mathcal{N}_S}$, $\abs{\mathcal{N}_A}$, $\abs{\mathcal{N}_B}\leq 2^k$, with $k$ to be set later. Let $\ket{\varphi}\in\mathcal{N}_S$, $U_A\in\mathcal{N}_A$, and $U_B\in\mathcal{N}_B$ be the elements with indices $x'\in\{0,1\}^k$, $y'\in\{0,1\}^k$ and $\gamma\in\{0,1\}^k$, respectively. We define $g$ as $g(x,y,\gamma)=0$ if $U\otimes V\ket{\varphi}$ is closer to $\mathcal{S}^{\varepsilon}_{\theta}$ than to $\mathcal{S}^{\varepsilon}_{\Bar\theta}$ in purified distance and $g(x,y,\gamma)=1$ if $U\otimes V\ket{\varphi}$ is closer to $\mathcal{S}^{\varepsilon}_{\Bar\theta}$ than to $\mathcal{S}^{\varepsilon}_{\theta}$ in purified distance. If neither is the case, we make the arbitrary choice $g(x,y,\gamma)=1$. By the assumption on $\varepsilon$,  $\mathcal{S}^\varepsilon_{\theta}\cap \mathcal{S}^\varepsilon_{\theta}=\emptyset$, and thus $g$ is well-defined.

We are going to show that $g$ is an $(\varepsilon,q)$-classical rounding. Consider an arbitrary $f:\{0,1\}^{2n}\rightarrow\{0,1\}$ and an arbitrary state $\ket\chi$ on $2q+2m_0$ qubits. Let $\ket{\chi}$, $\{U_A^x,U_B^y\}_{xy}$ be from a $q$-qubit strategy for \QPVcohfm, and choose $\gamma$, $f_A(x)$ and $f_B(y)$ to be the closest elements to $\ket{\chi}$, $U_A^x$ and $U_B^y$, respectively, in their corresponding $\delta$-nets in the Euclidean and operator norm, respectively, (if not unique, make an arbitrary choice) and let $\ket{\varphi},U_A,U_B$ be their corresponding elements. Assume $U_A^x\otimes U_B^y\ket{\chi}\in\mathcal{S}_{\theta}^{\varepsilon}$. Then,
\begin{equation}\label{eq distance in the net}
\begin{split}
    \mathcal{P}(&U_A^x\otimes U_B^y\ket{\chi},U_A\otimes U_B\ket{\varphi})\leq \| U_A^x\otimes U_B^y\ket{\chi}-U_A\otimes U_B\ket{\varphi}\|_2\\&\leq \|(U_A+U_A^x-U_A)\otimes(U_B+U_B^y-U_B)(\ket{\varphi}+\ket{\chi}-\ket{\varphi})-U_A\otimes U_B\ket{\varphi}\|_2
    \\&\leq3\delta+3\delta^2 +\delta^3< \frac{\Tilde\varepsilon}{2},
\end{split}
\end{equation}

where in the first inequality, we have used Lemma \ref{lemma unit vectors trace distance norm}, in the second, we have used the triangle inequality and the inequality $\|X\otimes Y \ket{x}\|_2\leq \|X\|_{\infty}\|Y\|_{\infty}\| \ket{x}\|_2$, together with $\|U_A^x-U_A\|_{\infty}, {\|U_B^y-U_B\|_{\infty}},\|\ket{\chi}-\ket{\varphi}\|\leq \delta$.
Thus, $U_A\otimes U_B\ket{\varphi}$ is closer to $\mathcal{S}_\theta^{\varepsilon}$ than to $\mathcal{S}_{\Bar\theta}^{\varepsilon}$.\\
Consider an $(\varepsilon, l)$-perfect strategy for \QPVcohfm~ and let $(x,y)$ be such that  $ h(U_{\theta}|PAB_c)_{\varphi},h(U_{\theta}| A_cB)_{\varphi}\leq h(U|P)_{\psi}+\varepsilon$ for $f(x,y)=0$. In particular, we have that $U_A^x\otimes U_B^y\ket{\chi}\in\mathcal{S}_i^{\varepsilon}$, and because of \eqref{eq distance in the net}, $f(x,y)=g(f_A(x),f_B(y),\gamma)$. Since there are at least $l$ pairs $(x,y)$ fulfilling it, $f(x,y)= g(f_A(x),f_B(y),\gamma)$ holds on at least $l$ pairs $(x,y)$ and therefore $g$ is an $(\varepsilon,q)$-classical rounding. The size of $k$ follows from Lemma \ref{lemma size delta-net}.
\end{proof}

\begin{lemma}
\label{lem:stable_E4}
Let $\varepsilon \in [0,1]$, $E,t,u>0$ be such that there exist $\Tilde{\varepsilon}>0$ and $\alpha$ such that \eqref{eq:condition_alpha_epsilon} holds. Let  $k$, $q \in \mathbb N$, $n =\Omega(m_0)$. Moreover, fix an $(\varepsilon, q)$-classical rounding $g$ of size $k$ with  
   $k=2^{2q+2m_0}\log\big(1+\frac{4}{\sqrt[3]{4(2+\Tilde\varepsilon)}-2}\big)$.
Let 
   $q = O(n-m_0)$. Then, with probability   $1-O(\lambda^{2^{m_0}})$ 
   
   the following holds:

A  uniformly random $f: \{0,1\}^{2n} \to \{0,1\}$ fulfills the following with probability at least 
${1 - O(2^{-2^{n}})}:$

 For any $f_A:\{0,1\}^n \to \{0,1\}^{k}$, $f_B:\{0,1\}^n \to \{0,1\}^{k}$, $\gamma \in \{0,1\}^{k}$, the equality $g(f_A(x), f_B(y), \gamma) = f(x,y)$ holds on less than $3/4$ of all pairs $(x,y)$.
\end{lemma}

\begin{proof}
     We want to estimate the probability that for a randomly chosen $f$, we can find $f_A$ and $f_B$ such that the corresponding function $g$ fulfils $\mathbb{P}_{x,y}[f(x,y)= g(f_A(x),f_B(y),\gamma)]\geq3/4$. In a similar manner as in \cite{bluhm2022single}, we have that
\begin{equation}\label{eq:prob<=}
     \mathbb{P}[f:\exists f_A,f_B,\gamma \textrm{ s.t. } \mathbb{P}_{x,y}[f(x,y) = g(f_A(x),f_B(y),\gamma)]]\leq
2^{ (2^{n+1}+1)k} 2^{2^{2n}h(1/4)} 2^{-2^{2n}},
\end{equation}
where $h$ denotes the binary entropy function. If $q=O(n-m_0)$ and $k=2^{2q+2m_0}\log\big(1+\frac{4}{\sqrt[3]{4(2+\Tilde\varepsilon)}-2}\big)$, the above expression is strictly upper bounded by $O(2^{-2^n})$. 
\end{proof}

In order to have explicit expressions instead of $n =\Theta(m_0)$ and $q = O(n-m_0)$, we have to fix the value of $\Tilde\varepsilon$. To obtain better bounds, we are interested in picking $\Tilde{\varepsilon}$ as large as possible. Given parameters $E,t,u$ and $\varepsilon$, known by the verifiers, we will be interested in picking values of $\alpha$ such that \eqref{eq:condition_alpha_epsilon} holds for $\Tilde\varepsilon$  as large as possible. This needs to be done numerically, since \eqref{eq:condition_alpha_epsilon} leads to a transcendental equation, see Section~\ref{section:parameters} for this analysis. This applies as well for the below theorem, which in short states that if the number of qubits the attackers pre-share at the beginning of the protocol, with high probability at least one of them will have a finite gap to the uncertainty of the honest prover regarding the value of the random variable $U$. That is, at least one attacker will have strictly larger uncertainty than the prover.

\begin{theorem} \label{thm:main} Let $\varepsilon \in [0,1]$, $E,t,u>0$ (under the control of the verifiers) be such that there exists $\Tilde{\varepsilon}>0$ and $\alpha$ such that \eqref{eq:condition_alpha_epsilon} holds. Let $n =\Theta(m_0)$. Let the number of qubits that Alice and Bob each control at the beginning of the protocol be  
\begin{equation}
    q=O(n-m_0).
\end{equation}
Then, with probability $1-O(\lambda^{2^{m_0}})$ 
the following holds.
   A random function $f$  fulfills the following with probability at least ${1 - O(2^{-2^{n}})}:$ the uncertainties for Alice and Bob when attacking the protocol \QPVcohf~are such that
    \begin{equation}\label{equ:att_h_lb}
    \max\{h(U_{\theta}|PAB_c)_{\phi},h(U_{\theta}| A_cB)_{\phi}\}\geq h(U|P)_{\psi}+\frac{\varepsilon}{4},
\end{equation}  
for every state $\ket{\phi}\in\mathbb C^{2q+2m_0}$, for $\theta\in\{0,\frac{\pi}{2}\}$.
\end{theorem}

\begin{proof}
   By Lemma \ref{lem:stable_E4}, with probability at least $1 - O(2^{-2^{n}})$, the function $f$ is such that there are no $(\varepsilon,\frac{3}{4}2^{2n})$-perfect $q$-qubit strategies for \QPVcohfm. That means that for every strategy, on a fraction at least $\frac{1}{4}$ of $(x,y)$, either $h(U_{\theta}|PAB_c)_{\phi}\geq h(U|P)_{\psi}+\frac{\varepsilon}{4}$ or $h(U_{\theta}| A_cB)_{\phi}\}\geq h(U|P)_{\psi}+\frac{\varepsilon}{4}$. 
\end{proof}

\subsection{Bounding attack success probability after repeated i.i.d.\ rounds}
For the following, remember that $\sigma \gg 1$, equivalently $\lambda \rightarrow1$. To estimate the number of (independent) rounds $N$ we have to run for the attack success probability to become vanishingly small, we cannot assume a specific attack distribution, and we have to assume the attackers have access to an ideal channel. By eq.~\eqref{equ:att_h_lb} we know that
\begin{align}
    h(U_{\theta}|E)_{\phi} &\coloneqq \max\{h(U_{\theta}|PAB_c)_{\phi},h(U_{\theta}| A_cB)_{\phi}\}\geq h(U|P)_{\psi}+\frac{\varepsilon}{4} \\
    &= \frac{1}{2} \log \left( \pi e \frac{1/2+u}{t} \right) + \frac{\varepsilon}{4}.
\end{align}
Now re-substituting $R = \sqrt{2}\lambda U$ yields
\begin{align}
    h(R|E)_{\phi} &\geq h(R|P)_{\psi}+\frac{\varepsilon}{4} \\
    &= \frac{1}{2} \log \left( 2\pi e \frac{1/2+u}{t} \right) + \frac{\varepsilon}{4} \\
    &\geq \frac{1}{2} \log \left( \pi e \right) + \frac{\varepsilon}{4},
\end{align}
where the last equation follows from eq.~\eqref{equ:h_honest} and the lower bound is smallest for the ideal channel with $t=1$ and $u=0$, which we assume attackers can use. Via the continuous variable version of Fano's inequality, Theorem~\ref{theorem:fano}, we can straightforwardly convert this into a lower bound for the estimation error of the attackers. We obtain
\begin{align}
    \mathbb{E}\left[ \left( R - r' \right)^2 \right] \geq \frac{1}{2\pi e} e^{2 h_\mathrm{nats}(R|E)_{\phi}} = \frac{1}{2} e^{\varepsilon/2}.
\end{align}
Thus $\mathbb{E}(\sqrt{t}R-r')^2 \geq \frac{1}{2} e^{\varepsilon/2}$ for any transmission $t$. The probability that the attackers' score falls below the threshold $\gamma$ is at most the probability that the score differs from $\mathbb{E}(\sqrt{t}R-r')^2/(1/2+u)$ by more than the difference $\Delta \coloneqq \frac{1/2}{1/2+u}e^{\varepsilon/2} - \gamma$.  Let $(r'^{\mathrm{att}}_i)_{i=1}^N$ be the sample of the attackers after $N$ i.i.d. rounds. Then we can then use the Chebyshev inequality for the random variable of the score to get 
\begin{align}
    \mathbb{P}\left[ \left| \frac{1}{N} \sum_{i=1}^N \frac{(\sqrt{t}r_i - r'^{\mathrm{att}}_i)^2}{1/2+u} - \frac{\mathbb{E} (\sqrt{t}R-r')^2}{1/2+u} \right| \geq \Delta \right] \leq \frac{\Tilde{\sigma}^2}{N\Delta^2} = O \left( \frac{1}{N \Delta^2} \right), 
\end{align}
where $\Tilde{\sigma}^2 = \mathbb{V} \left[ \frac{\left(\sqrt{t}R - r' \right)^2}{1/2+u} \right]$ is the variance. We set the tolerance for the attack success probability to $\varepsilon_\mathrm{att} = \varepsilon_\mathrm{hon}$ for simplicity. If we then set $N \Delta^2 = \Omega \left( \frac{1}{\varepsilon_{\mathrm{hon}}} \right)$, we get
\begin{align}
    \mathbb{P}\left[ \frac{1}{N} \sum_{i=1}^N \frac{(\sqrt{t}R_i - r_i')^2}{1/2+u} \leq \gamma \right] \leq O(\varepsilon_{\mathrm{hon}}).
\end{align}
The required number of rounds $N$ can be obtained by first setting the tolerated $\varepsilon_\mathrm{hon}$ and then solving $N \Delta^2 = \Omega \left( \frac{1}{\varepsilon_{\mathrm{hon}}} \right)$ for $N$. This means we accept the honest prover with probability at least $1-\varepsilon_{\mathrm{hon}}$, while accepting attackers with probability at most $\varepsilon_\mathrm{hon}$ after $N$ i.i.d. rounds.

\section{Concrete linear bounds for given experimental parameters}\label{section:parameters}

In the above section, we proved that, if $V_0$ prepares a two-mode squeezed state with squeezing parameter $\zeta$ and $\lambda=\tanh\zeta$, with probability $1-O(\lambda^{2^{m_0}})$ attackers who pre-share $q=O(n-m_0)$ qubits will not be able to mimic arbitrarily close an honest prover. For that, we need that  $\varepsilon,$ $E,t$ and $u$ are such that exists $\Tilde{\varepsilon}>0$ and $\alpha$ for which \eqref{eq:condition_alpha_epsilon} holds. In order to have parameters fulfilling \eqref{eq:condition_alpha_epsilon} for $\varepsilon>0$ we need that the attenuation parameter $t$ and the excess noise power $u$ of the quantum channel connecting $V_0$ and $P$ fulfill the relation \eqref{eq:relation_t_u}. In the section we analyze perfect and imperfect channels. 

\subsection{Perfect channel}
We start by considering a perfect channel connecting $V_0$ and $P$, given by $t=1$ and $u=0$. We fix $\varepsilon=0.1$ and assume the protocol is played enough rounds to statistically distinguish the honest party with an uncertainty
\begin{equation}
    h(U|P)_{\psi}\rightarrow\frac{1}{2}\log \frac{\pi e }{2}\simeq 1.0471,
\end{equation}
from the uncertainty of at least one of the attackers being lower bounded by
\begin{equation}
    \frac{1}{2}\log \frac{\pi e }{2}+\frac{\varepsilon}{4}\simeq 1.0721.
\end{equation}
Fix an energy bound $E=10^3$ in units such that $\hbar\omega=1$. Then, the largest $\Tilde{\varepsilon}$ that fulfills \eqref{eq:condition_alpha_epsilon} is 
\begin{equation}
    \Tilde\varepsilon\simeq0.0037,
\end{equation} for $\alpha\simeq0.036$, see Fig.~\ref{fig:alpha_epsilon} for a representation of the inequality \eqref{eq:condition_alpha_epsilon} where the value of $\varepsilon$ is fixed and represented as the black plane and the gray surface represents the right-hand side of the inequality.

\begin{figure}
    \centering
    \includegraphics[width=120mm]{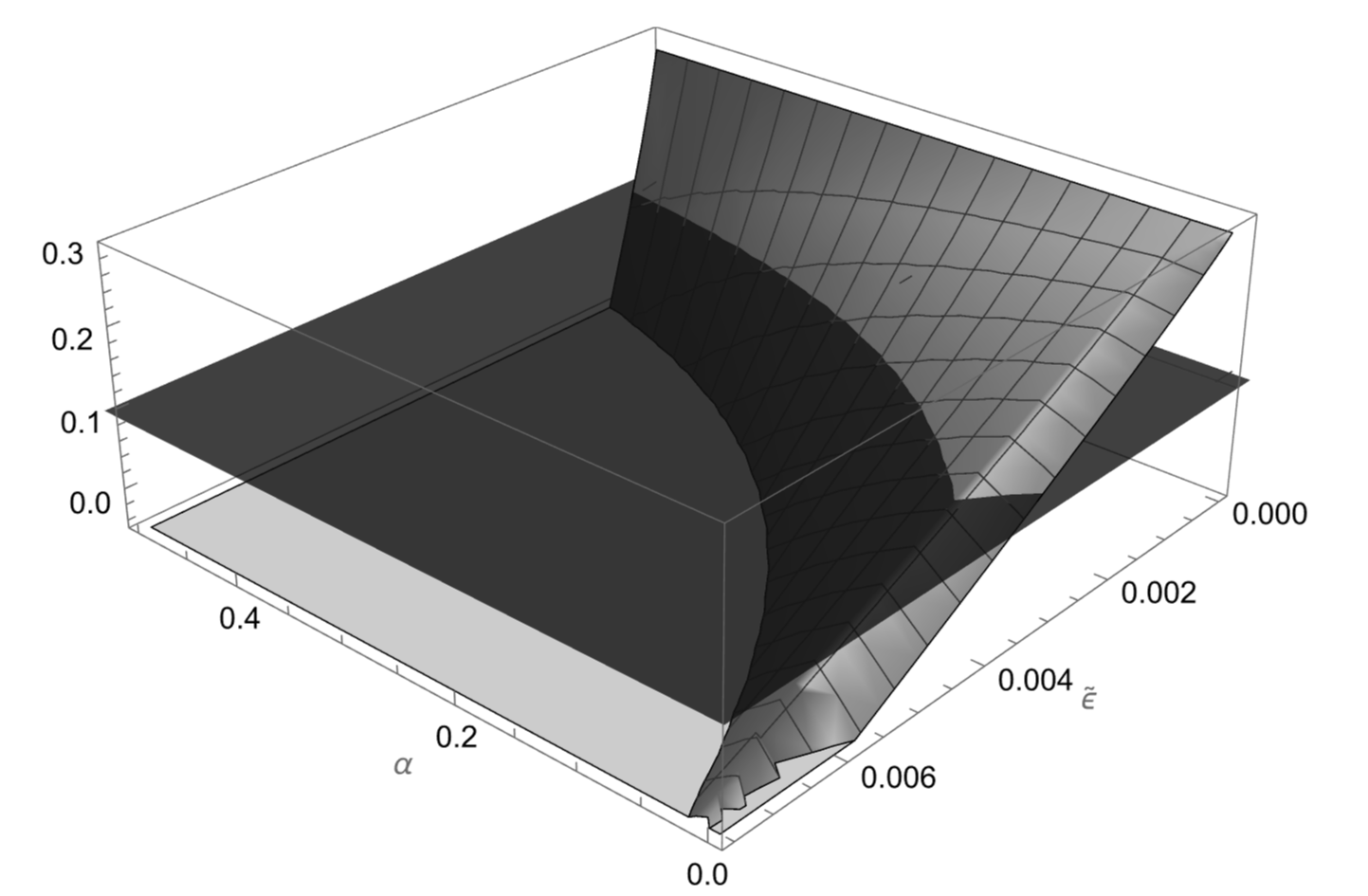}
    \caption{Representation of the left (transparent black) and right (gray surface) hand sides of the inequality \eqref{eq:condition_alpha_epsilon} for $\varepsilon=0.1, E=10^3, t=1$ and $u=0$. The region where the gray surface is above the black plain gives the $(\alpha,\Tilde{\varepsilon})$ such that the inequality \eqref{eq:condition_alpha_epsilon} holds.}
    \label{fig:alpha_epsilon}
\end{figure}

Notice that the energy term in \eqref{eq:condition_alpha_epsilon} scales as $\Tilde\varepsilon \log(E+1)$, i.e. logarithmically in $E$ with a small factor $\Tilde\varepsilon$ in front. Therefore, the inequality remains very stable with respect to $E$. For instance, if one picks $E=10,10^2,10^4$, the values of the maximum $\Tilde{\varepsilon}$ remain almost unchanged. 

For the values of $\varepsilon=0.1$ and $\Tilde{\varepsilon}=0.004$, the size of the $(\varepsilon,q)$-classical rounding in Proposition~\ref{prop:size_k} is $k=12\cdot2^{2q+2m_0}$. Then, \eqref{eq:prob<=} in the proof of Proposition~\ref{prop:size_k} is strictly upper bounded by $2^{-2^n}$ if $q\leq\frac{n}{2}-m_0-5$, for $n>2(m_0+5)$. Then, for the above energies, Theorem~\ref{thm:main} can be restated as follows.
\begin{corollary}\label{corollary:perfect_channel} Let $\varepsilon=0.1$, $t=1,u=0$. Let $n >2(m_0+5)$. Let the number of qubits that each Alice and Bob control at the beginning of the protocol be  
\begin{equation}
    q\leq\frac{n}{2}-m_0-5.
\end{equation}
Then, with probability $1-O(\lambda^{2^{m_0}})$ the following holds.
  A random function $f$  fulfills the following with probability at least $1 - 2^{-2^{n}}$: the uncertainties for Alice and Bob when attacking the protocol \QPVcohf~are such that
    \begin{equation}
    \max\{h(U_{\theta}|PAB_c)_{\phi},h(U_{\theta}| A_cB)_{\phi}\}\geq h(U|P)_{\psi}+\frac{\varepsilon}{4},
\end{equation}  
for every state $\ket{\phi}\in\mathbb C^{2q+2m_0}$, for $\theta\in\{0,\frac{\pi}{2}\}$.    
\end{corollary}

\subsection{Imperfect channel}
We do the analysis for an imperfect channel for some $(u,t)$ such that condition \eqref{eq:relation_t_u} is fulfilled. We pick the parameters plotted in Fig.~\ref{fig:u_t_relation}. For the values of $\varepsilon$, $t$ and $u$ in Table~\ref{Table:values_epsilon}, we find the maximum $\Tilde{\varepsilon}$ for $E=10^3$, and we have the same linear bounds as in Corollary~\ref{corollary:perfect_channel}.

\begin{table}[h]
\centering
\begin{tabular}{|c|c|c|c|c|}
\hline
$\varepsilon$ & $t$  & $u$   & $\alpha$ & $\Tilde{\varepsilon}$ \\ \hline
0.03          & 0.8  & 0.05  & 0.013    & 0.00031              \\ \hline
0.03          & 0.9  & 0.12  & 0.013    & 0.00029              \\ \hline
0.07          & 0.95 & 0.075 & 0.025    & 0.00131              \\ \hline
\end{tabular}
\caption{Maximum value of $\Tilde{\varepsilon}$ fulfilling \eqref{eq:condition_alpha_epsilon} given $\varepsilon, t $ and $u$ with its corresponding value of $\alpha$ that attains it.}
\label{Table:values_epsilon}
\end{table}

\section{Open problems}
In the discrete variable case \cite{allerstorfer2023making} has recently found a way around the problem of transmission loss. It's an interesting open question whether the idea in \cite{allerstorfer2023making} also could work for CV-QPV to make the practical appeal of studying these protocols higher.

\subsubsection*{Acknowledgments}
We thank Philip Verduyn Lunel for valuable discussions. RA was supported by the Dutch Research Council (NWO/OCW), as part of the Quantum Software Consortium programme (project number 024.003.037). FS and LEF are supported by the Dutch Ministry of Economic Affairs and Climate Policy (EZK), as part of the Quantum Delta NL programme. B\v{S} and AAR acknowledge the support from Groeifonds Quantum Delta NL KAT2.

\bibliographystyle{alphaurl}
\bibliography{references}
\end{document}